\documentclass[11pt,leqno]{article}
\usepackage{amssymb,amsfonts}
\usepackage{amsmath,amsthm,amsxtra}
\usepackage[all]{xy}
\usepackage{color}
\usepackage{verbatim}

\usepackage{enumerate}

\newcommand\bigcheck[1]{#1 \raise1ex\hbox{$\hspace{-1ex}{}^\vee$}}
\newcommand\sucheck[1]{#1 \raise0.5ex\hbox{$\hspace{-1ex}{}^\vee$}}

\makeatletter
\@addtoreset{equation}{section}
\makeatother

\newtheorem{theorem}{Theorem}[section]
\newtheorem{lemma}[theorem]{Lemma}
\newtheorem{corollary}[theorem]{Corollary}
\newtheorem{proposition}[theorem]{Proposition}
\newtheorem*{lemma*}{Lemma}

\theoremstyle{definition}
\newtheorem{definition}[theorem]{Definition}

\theoremstyle{remark}
\newtheorem{remark}[theorem]{Remark}

\newcommand{\mc}[1]{{\mathcal #1}}

\newcommand{\mb}[1]{{\mathbb #1}}

\newcommand\tint{{\textstyle\int}}

\renewcommand{\tilde}{\widetilde}

\newcommand{\ad}{\mathop{\rm ad }}

\newcommand{\Mat}{\mathop{\rm Mat }}

\renewcommand{\ker}{\mathop{\rm Ker }}

\newcommand{\Span}{\mathop{\rm Span }}

\definecolor{light}{gray}{.9}

\begin{document}


\title{A new approach to the Lenard-Magri scheme of integrability}

\author{
Alberto De Sole
\thanks{Dipartimento di Matematica, Universit\`a di Roma ``La Sapienza'',
00185 Roma, Italy ~~
and IHES, Bures sur Yvette, France~~
desole@mat.uniroma1.it ~~~~
Supported in part by PRIN and FIRB grants.
},~~
Victor G. Kac
\thanks{Department of Mathematics, M.I.T.,
Cambridge, MA 02139, USA.~~
and IHES, Bures sur Yvette, France~~
kac@math.mit.edu~~~~
Supported in part by Simons Fellowship.
}~~
and Refik Turhan 
\thanks{
Department of Engineering Physics,
Ankara University,
06100 Tandogan, Ankara
Turkey.~~
turhan@eng.ankara.edu.tr
}
}

\maketitle

\begin{abstract}
We develop a new approach to the Lenard-Magri scheme of integrability
of bi-Hamiltonian PDE's, when one of the Poisson structures is a strongly skew-adjoint 
differential operator.
\end{abstract}

\section{Introduction}\label{sec:1}

There have been a number of papers on classification and study of integrable PDE's
in the past few decades.
As a result, the 1-component integrable PDE's have been to a large extent classified.
%
In the 2-component case there have been only partial results,
see the survey \cite{MNW09} and references there.

In the present paper we prove integrability of the 2-component PDE's
\eqref{20130306:eq1}, \eqref{20130306:eq2}, \eqref{20130306:eq4}, \eqref{20130306:eq3b}.
All these equations enter in the same bi-Hamiltonian hierarchy of PDE's.
The corresponding two compatible Poisson structures
given by \eqref{20130307:eq1} and \eqref{20130307:eq2}
are local and have order 3 and 5.
Equation \eqref{20130306:eq4} appeared in \cite{MNW07}.
%

The proof of integrability (i.e. existence of integrals of motion in involution
with the associated Hamiltonian vector fields of arbitrarily high order),
uses the Lenard-Magri scheme.
Unfortunately (or fortunately) the existing methods, developed 
in \cite{Dor93}, \cite{BDSK09}, \cite{Wan09},
do not quite work here.
We therefore develop a new method based on the notion of strongly skew-adjoint
differential operators,
using the Lie superalgebra of variational polyvector fields.

We show that the Lenard-Magri scheme ``almost'' always works
provided that one of the Poisson structures $H_0$ 
is a non-degenerate strongly skew-adjoint operator.
Namely, we are able to show, by a general argument, 
that the Poisson brackets of conserved densities are Casimir elements for $H_0$,
but then we need to check by simple differential order considerations
that these brackets are actually zero.


Throughout the paper, unless otherwise specified,
all vector spaces are considered over a field $\mb F$ of characteristic zero.

\section{A Lemma on $\mb Z$-graded Lie superalgebras}\label{sec:2}

Let $W=W_{-1}\oplus W_0\oplus W_1\oplus\dots$ be a $\mb Z$-graded Lie superalgebra
such that elements in $W_j,\,j\geq-1$, have parity $j$ mod $2\in\mb Z/2\mb Z$.
Given $H\in W_1$, we define the following skew-symmetric bracket $\{\cdot\,,\,\cdot\}_H$
on $W_{-1}$ considered as an even space:
\begin{equation}\label{20130222:eq1}
\{f,g\}_H=[[H,f],g]
\,\,,\,\,\,\,
f,g\in W_{-1}\,.
\end{equation}
Recall that if $[H,H]=0$, then $\{\cdot\,,\,\cdot\}_H$ is actually a Lie bracket
on $W_{-1}$, considered as an even space (see e.g. \cite{DSK13}).
We also define the space of \emph{Casimirs} for $H$ as
\begin{equation}\label{20130222:eq8}
C_{-1}(H)\,:=\,\big\{f\in W_{-1}\,\big|\,[H,f]=0\big\}\subset W_{-1}\,.
\end{equation}
\begin{lemma}\label{20130222:lem}
\begin{enumerate}[(a)]
\item
Let $H_0,H_1\in W_1$. Denote by $\{\cdot\,,\,\cdot\}_0$ and $\{\cdot\,,\,\cdot\}_1$
the corresponding brackets on $W_{-1}$ given by \eqref{20130222:eq1}.
Let $f_0,f_1,\dots,f_{N+1}\in W_{-1}$ satisfy the equations
\begin{equation}\label{20130222:eq2}
[H_1,f_n]=[H_0,f_{n+1}]
\,\,\text{ for all }\,\,
 n=0,\dots,N\,.
\end{equation}
Then we have
$$
\{f_m,f_n\}_0=\{f_m,f_n\}_1=0
\,\,\text{ for all }\,\,
m,n=0,\dots,N+1\,.
$$
\item
Let $H_0,H_1\in W_1$.
Let $f_0,f_1,\dots,f_{N+1}\in W_{-1}$, with $f_0\in C_{-1}(H_0)$,
satisfy equations \eqref{20130222:eq2},
and let $\{g_n\}_{n\in\mb Z_+}\subset W_{-1}$ satisfy $[H_1,g_n]=[H_0,g_{n+1}]$
for all $n\in\mb Z_+$.
Then we have
$$
\{f_m,g_n\}_0=\{f_m,g_n\}_1=0
\,\,\,\,
\text{ for all } m=0,\dots,N+1,\,n\in\mb Z_+\,.
$$
\item
If $H\in W_1$ is such that $[H,H]=0$,
then $\ad H$ defines a Lie algebra homomorphisms $(W_{-1},\{\cdot\,,\,\cdot\}_H)\to(W_0,[\cdot\,,\,\cdot])$,
i.e.
$$
[H,\{f,g\}_H]=[[H,f],[H,g]]
\,\,\,\,
\text{ for all } f,g\in W_{-1}\,.
$$
\item
Let $H_0,H_1\in W_1$ be such that $[H_0,H_1]=0$.
Then
$C_{-1}(H_0)\subset V$ is a closed with respect to the bracket $\{\cdot\,,\,\cdot\}_1$.
\item
Let $H_0,H_1\in W_1$ be such that $[H_0,H_1]=[H_1,H_1]=0$.
Suppose that $f_0,f_1,\dots,f_{N+1}\in W_{-1}$ satisfy equations \eqref{20130222:eq2}.
Then
$\{f_{n+1},g\}_1\in C_{-1}(H_0)$ for all $n=0,\dots,N$
and $g\in C_{-1}(H_0)$.
\end{enumerate}
\end{lemma}
\begin{proof}
First, we prove part (a).
By skew-symmetry, we have $\{f_n,f_n\}_0=\{f_n,f_n\}_1=0$ for every $n$.
Assuming that $n>m$, we prove, by induction on $n-m$ that
$\{f_m,f_n\}_0=\{f_m,f_n\}_1=0$.
We have
$$
\{f_m,f_n\}_1=[[H_1,f_m],f_n]=[[H_0,f_{m+1}],f_n]=\{f_{m+1},f_n\}_0\,,
$$
which is zero by inductive assumption.
Similarly,
$$
\begin{array}{l}
\displaystyle{
\{f_m,f_n\}_0=-\{f_n,f_m\}_0=-[[H_0,f_n],f_m]
} \\
\displaystyle{
=-[[H_1,f_{n-1}],f_m]=-\{f_{n-1},f_m\}_1=\{f_m,f_{n-1}\}_1\,,
}
\end{array}
$$
which again is zero by induction. 

Next, we prove part (b). 
Since, by assumption, $f_0\in C_{-1}(H_0)$, we have
$\{f_0,g_n\}_0=[[H_0,f_0],g_n]=0$ for all $n$.
Furthermore,
$$
\{f_0,g_n\}_1=-\{g_n,f_0\}_1=-[[H_1,g_n],f_0]=-[[H_0,g_{n+1}],f_0]=\{f_0,g_{n+1}\}_0\,,
$$
which is zero by the previous case.
We next prove, by induction on $m\geq1$, that
$\{f_m,g_n\}_0=\{f_m,g_n\}_1=0$ for every $n\in\mb Z_+$.
We have
$$
\{f_m,g_n\}_0=[[H_0,f_m],g_n]=[[H_1,f_{m-1}],g_n]=\{f_{m-1},g_n\}_1\,,
$$
which is zero by inductive assumption, and
$$
\begin{array}{l}
\displaystyle{
\{f_m,g_n\}_1=-\{g_n,f_m\}_1=-[[H_1,g_n],f_m]=-[[H_0,g_{n+1}],f_m]
} \\
\displaystyle{
=[[H_0,f_m],g_{n+1}]=\{f_m,g_{n+1}\}_0
\,,
}
\end{array}
$$
which is zero by the previous case, completing the proof of part (b).

For $H\in W_1$ and $f,g\in W_{-1}$, we have, by the Jacobi identity,
$$
[H,\{f,g\}_H]
=[H,[[H,f],g]]
=[[H,[H,f]],g]+[[H,f],[H,g]]
\,.
$$
If $[H,H]=0$, we have $[H,[H,f]]=0$, since $H$ is odd,
proving part (c).

Next, we prove part (d).
If $f,g\in C_{-1}(H_0)$, we have
$$
[H_0,\{f,g\}_1]=[H_0,[[H_1,f],g]]\,,
$$
and this is zero since, by assumption, $H_0$ commutes with all elements $H_1$, $f$ and $g$.

Finally, we prove part (e).
Since, by assumption, $H_0$ commutes with both $H_1$ and $g$, we have,
by the Jacobi identity,
$$
\begin{array}{l}
\displaystyle{
[H_0,\{f_{n+1},g\}_1]
=[H_0,[[H_1,f_{n+1}],g]]
=[[H_0,[H_1,f_{n+1}]],g]
} \\
\displaystyle{
=-[[H_1,[H_0,f_{n+1}]],g]
=-[[H_1,[H_1,f_n]],g]
\,,}
\end{array}
$$
and this is zero since $(\ad H_1)^2=0$.
\end{proof}

\section{Application to the theory of Hamiltonian PDE's}\label{sec:3}

In the present paper we will use Lemma \ref{20130222:lem}
in the special case when $\mc W$ is the Lie superalgebra of variational polyvector fields
over an algebra of differential functions $\mc V$.

Recall from \cite{BDSK09}
that an algebra of differential function $\mc V$ in the variables $u_i,\,i\in I=\{1,\dots,\ell\}$,
is a differential algebra extension of the algebra of differential polynomials
$R_\ell=\mb F[u_i^{(n)}\,|\,i\in I,n\in\mb Z_+]$,
with the ``total derivative'' $\partial$ defined on generators by $\partial u_i^{(n)}=u_i^{(n+1)}$,
and endowed with commuting derivations $\frac{\partial}{\partial u_i^{(n)}}:\,\mc V\to\mc V$
extending the usual partial derivatives on $R_\ell$,
such that for every $f\in\,\mc V$ we have $\frac{\partial f}{\partial u_i^{(n)}}=0$
for all but finitely many values of $i$ and $n$,
and satisfying the commutation relation 
$[\frac{\partial}{\partial u_i^{(n)}},\partial]=\frac{\partial}{\partial u_i^{(n-1)}}$
for every $i\in I,n\in\mb Z_+$ (the RHS is considered to be 0 for $n=0$).

Recall from \cite{DSK13} some properties of the Lie superalgebra $\mc W$ 
of variational polyvector fields over $\mc V$ that we will need.
It is a $\mb Z$-graded Lie superalgebra $\mc W=\mc W_{-1}\oplus\mc W_0\oplus\mc W_1\oplus\dots$,
with the parity compatible with the $\mb Z$-grading.
Furthermore, 
$\mc W_{-1}=\mc V/\partial\mc V$ is the space of local functionals,
$\mc W_0=\mc V^\ell$ is the Lie algebra of evolutionary vector fields,
i.e. derivations of $\mc V$ commuting with $\partial$,
which have the form $X_P=\sum_{i\in I,n\in\mb Z_+}(\partial^n P_i)\frac{\partial}{\partial u_i^{(n)}}$, 
$P\in\mc V^\ell$.
The bracket of two evolutionary vector fields is given by the formula $[X_P,X_Q]=X_{[P,Q]}$,
where
$$
[P,Q]=D_Q(\partial)P-D_P(\partial)Q\,,
$$
and $D_P(\partial)\in\Mat_{\ell\times\ell}\mc V[\partial]$ 
is the \emph{Frechet derivative} of $P$:
\begin{equation}\label{20130222:eq5}
{D_P(\partial)}_{ij}=\sum_{n\in\mb Z_+}\frac{\partial P_i}{\partial u_j^{(n)}}\partial^n\,.
\end{equation}
The Lie bracket between elements $P\in W_0$ and $\tint f\in W_{-1}$ is given by
\begin{equation}\label{20130307:eq21}
[P,\tint f]=\tint X_P(f)=\tint P\cdot \delta f\,,
\end{equation}
where $\delta f=\big(\frac{\delta f}{\delta u_i}\big)_{i\in I}\in\mc V^\ell$ 
denotes the vector of variational derivatives of $\tint f\in\mc V/\partial\mc V$:
\begin{equation}\label{20130222:eq6}
\frac{\delta f}{\delta u_i}=
\sum_{n\in\mb Z_+}(-\partial)^n\frac{\partial f}{\partial u_i^{(n)}}\,.
\end{equation}
Finally, $\mc W_1$ is the space of skew-adjoint $\ell\times\ell$ matrix 
differential operators over $\mc V$,
the Lie bracket between $H\in\mc W_1$ and $\tint f\in\mc V/\partial\mc V=\mc W_{-1}$ is
given by
\begin{equation}\label{20130307:eq22}
[H,\tint f]=H(\partial)\delta f\,,
\end{equation}
and a \emph{Poisson structure} on $\mc V$ is an element $H\in\mc W_1$
such that $[H,H]=0$.

For $H\in\mc W_{1}$, the corresponding skew-symmetric bracket \eqref{20130222:eq1}
on $\mc W_{-1}=\mc V/\partial\mc V$ is given by the usual formula
\begin{equation}\label{20130222:eq7}
\{\tint f,\tint g\}_H
=\tint \delta g\cdot H(\partial)\delta f\,,
\end{equation}
and this bracket
defines a Lie algebra structure on $\mc V/\partial\mc V$
if and only if $H$ is a Poisson structure on $\mc V$.
In this context, 
the space \eqref{20130222:eq8} of Casimir elements for $H$ is
\begin{equation}\label{20130222:eq9}
C_{-1}(H)
=\Big\{\tint f\in\mc V/\partial\mc V\,\Big|\,H(\partial)\delta f=0\Big\}\,.
\end{equation}

Recall that the \emph{Hamiltonian partial differential equation} 
for the Poisson structure $H\in\mc W_1$
and the Hamiltonian functional $\tint h\in\mc V/\partial\mc V$ 
is the following evolution equation in the variables $u_1,\dots,u_\ell$:
\begin{equation}\label{20130222:eq10}
\frac{du_i}{dt}=\sum_{j=1}^\ell H_{ij}(\partial)\frac{\delta h}{\delta u_j}\,.
\end{equation}
An \emph{integral of motion} for the Hamiltonian equation \eqref{20130222:eq10}
is a local functional $\tint f\in\mc V/\partial\mc V$ such that $\{\tint h,\tint f\}_H=0$.
Equation \eqref{20130222:eq10} is said to be \emph{integrable}
if there is an infinite sequence of linearly independent integrals of motion 
$\tint h_0=\tint h,\tint h_1,\tint h_2,\dots$
in involution: $\{\tint h_m,\tint h_n\}_H=0$ for all $m,n\in\mb Z_+$.

One of the main techniques for proving integrability of a Hamiltonian equation
is based on the so called \emph{Lenard-Magri scheme of integrability}.
This applies when a given evolution equation has a bi-Hamiltonian form,
i.e. it can be written in Hamiltonian form in two ways:
\begin{equation}\label{20130222:eq11a}
\frac{du}{dt}=H_1(\partial)\delta h_0=H_0(\partial)\delta h_1\,,
\end{equation}
where $H_0,H_1$ are compatible Poisson structures on $\mc V$,
namely they satisfy $[H_0,H_0]=[H_0,H_1]=[H_1,H_1]=0$.
In this situation, the Lenard-Magri scheme consists in finding a sequence
of local functionals $\tint h_0,\tint h_1,\tint h_2,\dots$
satisfying the recursive conditions
\begin{equation}\label{20130222:eq11}
H_1(\partial)\delta h_n=H_0(\partial)\delta h_{n+1}\,,
\end{equation}
for all $n\in\mb Z_+$.
Lemma \ref{20130222:lem}(a) and (c)  guarantees that, in this situation,
all local functionals $\tint h_n,\,n\in\mb Z_+$ are integrals of motion
in involution with respect to both Poisson brackets
$\{\cdot\,,\,\cdot\}_0$ and $\{\cdot\,,\,\cdot\}_1$,
and the higher symmetries $P_n=H_0(\partial)\delta h_n$ commute.
Indeed, we have the following immediate consequences of Lemma \ref{20130222:lem}.
\begin{corollary}\label{20130222:cor}
Let $H_0,H_1$ be compatible Poisson structures on $\mc V$,
and let $\{\cdot\,,\,\cdot\}_0$ and $\{\cdot\,,\,\cdot\}_1$
be the corresponding brackets on $\mc V/\partial\mc V$ given by \eqref{20130222:eq7}.
Let $\{\tint h_n\}_{n\in\mb Z_+}\subset\mc V/\partial\mc V$ 
be a sequence of local functionals satisfying 
the Lenard-Magri recursive equations \eqref{20130222:eq11}.
Then
all elements $\tint h_n$ are integrals of motion for the bi-Hamiltonian equation \eqref{20130222:eq11a} 
in involution with respect to both Poisson brackets for $H_0$ and $H_1$:
$\{\tint h_m,\tint h_n\}_0=\{\tint h_m,\tint h_n\}_1=0$,
and all Hamiltonian vector fields $P_n=H(\partial)\delta h_n$ commute: 
$[P_m,P_n]=0$, for all $m,n\in\mb Z_+$.
\end{corollary}
\begin{proof}
The first statement is a special case of Lemma \ref{20130222:lem}(a),
and the second statement follows by Lemma \ref{20130222:lem}(c).
\end{proof}

The main problem in applying the Lenard-Magri scheme of integrability 
is to show that at each step $n$ the recursive equation \eqref{20130222:eq11}
can be solved for $\tint h_{n+1}\in\mc V/\partial\mc V$.
This problem is split in three parts.
First, under the assumption that $\mc V$ is a domain 
and the Poisson structure $H_0$ is non-degenerate
(cf. Definition \ref{20130222:def1} below),
Theorem \ref{20130222:thm1} below
guarantees that,
if an element $F\in\mc V^\ell$ exists such that 
$H_1(\partial)\delta h_n=H_0(\partial)F$,
then $F$ is \emph{closed}, i.e. it has self-adjoint Frechet derivative:
$D_{F}(\partial)^*=D_{F}(\partial)$.
Next,
Theorem \ref{20130222:thm2} below
shows that, if $F\in\mc V^\ell$ is closed,
then it is \emph{exact} in a normal extension $\tilde{\mc V}$
of the algebra of differential algebra function $\mc V$:
$F=\delta\tint h$ for some $h\in\tilde{\mc V}$.
Hence, we reduced our problem to proving that
$H_1(\partial)\delta h_n\in H_0(\partial)\mc V^\ell$.
There is no universal technique to solve this problem,
but there are various approaches which work in specific examples
(see e.g. \cite{BDSK09}, \cite{DSK13}, \cite{DSK12}, \cite{Dor93}, \cite{Olv93}, \cite{Wan09}).
In Proposition \ref{20130222:prop3} below
we propose an ansatz for solving this problem,
under the assumption that the Poisson structure $H_0$ is strongly skew-adjoint
(cf. Definition \ref{20130222:def3} below),
and that the given finite Lenard-Magri sequence $\tint h_0,\dots,\tint h_n$
starts  with a Casimir element for $H_0$: $H_0(\partial)\delta(\tint h_0)=0$.
In the following sections we will be able to apply successfully this ansatz
to prove integrability of the compatible bi-Hamiltonian
PDE's in two variables
\eqref{20130306:eq1}, \eqref{20130306:eq2}, \eqref{20130306:eq4} and \eqref{20130306:eq3b}.

\begin{definition}\label{20130222:def1}
Assume that $\mc V$ is a domain.
A matrix differential operator $H\in\Mat_{\ell\times\ell}\mc V[\partial]$
is \emph{non-degenerate} if 
it is not a left (or right) zero divisor in $\Mat_{\ell\times\ell}\mc V[\partial]$
(equivalently, if its Dieudonn\'e determinant in non-zero).
\end{definition}
\begin{theorem}[see e.g. {\cite[Thm.2.7]{BDSK09}}]\label{20130222:thm1}
Let $H_0,\,H_1\in\Mat_{\ell\times\ell}\mc V[\partial]$ be compatible Poisson structures
on the algebra of differential functions $\mc V$, which is assumed to be a domain,
and suppose that $H_0$ is non-degenerate.
If $\tint h_0,\tint h_1\in\mc V/\partial\mc V$ and $F\in\mc V^\ell$
are such that $H_1(\partial)\delta h_0=H_0(\partial)\delta h_1$
and $H_1(\partial)\delta h_1=H_0(\partial)F$,
then $F$ is closed: $D_F(\partial)=D_F(\partial)^*$.
\end{theorem}

Consider the following filtration of the algebra of differential functions $\mc V$:
$$
\mc V_{m,i}\,=\,
\Big\{ f\in\mc V\,\Big|\,\frac{\partial f}{\partial u_j^{(n)}}=0
\text{ for all } (n,j)>(m,i)
\Big\}\,,
$$
where $>$ denotes lexicographic order.
By definition, $\frac{\partial}{\partial u_j^{(n)}}(\mc V_{m,i})$
is zero for $(n,j)>(m,i)$, and it is contained in $\mc V_{m,i}$ for $(n,j)\leq(m,i)$.
\begin{definition}\label{20130222:def2}
The algebra of differential functions $\mc V$ is called \emph{normal}
if $\frac{\partial}{\partial u_i^{(m)}}(\mc V_{m,i})=\mc V_{m,i}$
for all $i\in I,m\in\mb Z_+$.
\end{definition}
Note that any algebra of differential function can be extended to a normal one
(see \cite{DSK13a}).
\begin{theorem}[{\cite[Prop.1.9]{BDSK09}}]\label{20130222:thm2}
If $F\in\mc V^\ell$ is exact, i.e. $F=\delta f$ for some $\tint f\in\mc V/\partial\mc V$,
then it is closed, i.e. $D_F(\partial)=D_F(\partial)^*$.
Conversely, 
if $\mc V$ is a normal algebra of differential functions
and $F\in\mc V^\ell$ is closed, then it is exact.
\end{theorem}

Note that if $H\in\Mat_{\ell\times\ell}\mc V[\partial]$ is a skew-adjoint operator,
then $H(\partial)\mc V^\ell\perp\ker H(\partial)$,
where the orthogonal complement is with respect to the pairing 
$\mc V^\ell\times\mc V^\ell\to\mc V/\partial\mc V$ given by $(F,P)\mapsto\tint F\cdot P$.
\begin{definition}\label{20130222:def3}
A skew-adjoint operator $H\in\Mat_{\ell\times\ell}\mc V[\partial]$ 
is called \emph{strongly skew-adjoint} if the following conditions hold:
\begin{enumerate}[(i)]
\item
$\ker H(\partial)\subset\delta(\mc V/\partial\mc V)$,
\item
$\big(\ker H(\partial)\big)^\perp=H(\partial)\mc V^\ell$.
\end{enumerate}
\end{definition}
Let $\tint h_0\in C_{-1}(H_0)$,
and let $\tint h_0,\dots,\tint h_N\in\mc V/\partial\mc V$ 
be a finite sequence satisfying the Lenard-Magri 
recursive equations \eqref{20130222:eq11}.
Lemma \ref{20130222:lem}(d) and (e), in this context, imply the following result:
\begin{corollary}\label{20130222:cor2}
$H_0(\partial)\delta\{\tint h_n,\tint g\}_1=0$
for all $n=0,\dots,N$ and for all $\tint g\in C_{-1}(H_0)$.
\end{corollary}
If $H_0$ is non-degenerate, its kernel in $\mc V^\ell$ is finite-dimensional.
Therefore it is reasonable to hope that,
by computing explicitly $\ker H_0$ and $C_{-1}(H_0)$,
and carefully looking at the recursive equation \eqref{20130222:eq11}
one can prove that, in fact, 
$\{\tint h_n,\tint g\}_1=0$ for all $\tint g\in C_{-1}(H_0)$.
In this case, assuming that $H_0$ is strongly skew-adjoint,
the following proposition guarantees that we can successfully apply the Lenard-magri scheme.
\begin{proposition}\label{20130222:prop3}
Suppose that $H_0$ is a strongly skew-adjoint operator 
and that $\{\tint h,\tint g\}_1=0$ for all $\tint g\in C_{-1}(H_0)$,
Then $H_1(\partial)\delta h\in H_0(\partial)\mc V^\ell$.
\end{proposition}
\begin{proof}
By assumption, we have $\tint \delta g\cdot H_1(\partial)\delta h=0$ for all $\tint g\in C_{-1}(H_0)$,
i.e. $H_1(\partial)\delta h\perp\delta C_{-1}(H_0)$.
By condition (i) of the strong skew-adjointness assumption on $H_0$
this implies that $H_1(\partial)\delta h\perp\ker H_0(\partial)$,
and therefore, by the condition (ii), we conclude that $H_1(\partial)\delta h\in H_0(\partial)\mc V^\ell$,
proving the claim.
\end{proof}

To conclude the section, 
we state the following result, which will be used in the following sections.
\begin{corollary}\label{20130222:cor4}
Let $H_0,H_1$ be compatible Poisson structures on $\mc V$.
Let $\{\tint f_n\}_{n=0}^N\subset\mc V/\partial\mc V$ 
be a finite sequence satisfying 
the Lenard-Magri recursive equations \eqref{20130222:eq11},
with $\tint f_0\in C_{-1}(H_0)$,
and let $\{\tint g_n\}_{n=0}^\infty\subset\mc V/\partial\mc V$ 
be an infinite sequence also satisfying 
the Lenard-Magri recursive equations \eqref{20130222:eq11}.
Then the two Lenard-Magri sequences are compatible, in the sense that
$\{\tint f_m,\tint g_n\}_0=\{\tint f_m,\tint g_n\}_1=0$ for all $m=0,\dots,N,\,n\in\mb Z_+$.
\end{corollary}
\begin{proof}
It is a special case of Lemma \ref{20130222:lem}(b).
\end{proof}

\section{The bi-Poisson structure $(H_0,H_1)$}\label{sec:4}

Consider the following algebra of differential functions in two variables $u,v$:
\begin{equation}\label{20130307:eq5}
\mc V=\mb F[u,v^{\pm1},u',v',u'',v'',\dots]\,.
\end{equation}
It is contained in the normal extension $\tilde{\mc V}=\mc V[\log v]$, 
see \cite[Ex.4.5]{DSK13a}.
\begin{theorem}\label{20130305:thm}
The following is a compatible pair of Poisson structures 
$(H_0,H_1)\in\Mat_{2\times2}\mc V[\partial]$:
\begin{equation}\label{20130307:eq1}
H_0(\partial)=
\left(\begin{array}{cc}
\partial^3+\partial\circ u+u\partial & v\partial \\
\partial\circ v & 0
\end{array}\right)
\,,
\end{equation}
and
\begin{equation}\label{20130307:eq2}
H_1(\partial)=
\left(\begin{array}{cc}
0 & \partial\circ\frac1{v^2} \\
\frac1{v^2}\partial & -\frac1{v^2}Q(\partial)\circ\frac1{v^2}
\end{array}\right)\,,
\end{equation}
where
$$
Q(\partial)=\partial^5+3\partial\circ(\partial\circ u+u\partial)\partial
+2(\partial^3\circ u+u\partial^3)+8(\partial\circ u^2+u^2\partial)\,.
$$
\end{theorem}
\begin{proof}
$H_0$ is a well known Poisson structure, see e.g. \cite{Ito82,Dor93}.
A simple proof of this fact can be found in \cite{BDSK09}.
$H_1$ is obviously skew-adjoint. 
The proof that $H_1$ satisfies Jacobi identity and is compatible with $H_0$
is a rather lengthy computation (one has to verify equation (1.49) from \cite{BDSK09}).
This has been checked with the use of the computer.
\end{proof}

\section{Casimirs for $H_0$ and $H_1$}\label{sec:5}

In this paper 
we will apply the Lenard-Magri scheme for the bi-Poisson structure $(H_0,H_1)$ 
to find integrable hierarchies of bi-Hamiltonian equations.
As explained in Section 1, 
in order to do so it is convenient to find the Casimir elements 
for $H_0$ and $H_1$.
\begin{proposition}\label{20130305:prop}
\begin{enumerate}[(a)]
\item
The kernel of $H_0(\partial)$ is spanned by
\begin{equation}\label{20130307:eq10}
\xi^{0,0}=\left(\begin{array}{c}
0 \\ 1
\end{array}\right)
\,\,,\,\,\,\,
\xi^{0,1}=\left(\begin{array}{c}
\frac1v \\ -\frac u{v^2}-\frac32\frac{(v')^2}{v^4}+\frac{v''}{v^3}
\end{array}\right)
\,.
\end{equation}
\item
We have
$\xi^{0,0}=\delta(\tint h^{0,0})$ and $\xi^{0,1}=\delta(\tint h^{0,1})$,
where
$$
\tint h^{0,0}=\tint v
\,\,\text{ and }\,\,
\tint h^{0,1}=\tint \Big(\frac uv-\frac12\frac{(v')^2}{v^3}\Big)\,.
$$
\item
The matrix differential operator $H_0(\partial)$ is strongly skew-adjoint.
\item
The kernel of $H_1(\partial)$ is spanned by
\begin{equation}\label{20130307:eq11}
\xi^{1,0}=\left(\begin{array}{c}
1 \\ 0
\end{array}\right)
\,\,,\,\,\,\,
\xi^{1,1}=\left(\begin{array}{c}
u''+4u^2 \\ \frac{v^2}2
\end{array}\right)
\,.
\end{equation}
\item
We have
$\xi^{1,0}=\delta(\tint h^{1,0})$ and $\xi^{1,1}=\delta(\tint h^{1,1})$,
where
$$
\tint h^{1,0}=\tint u
\,\,\text{ and }\,\,
\tint h^{1,1}=\tint\Big(\frac12uu''+\frac43 u^3+\frac16v^3\Big)\,.
$$
\item
The matrix differential operator $H_1(\partial)$ is strongly skew-adjoint.
\end{enumerate}
\end{proposition}
\begin{proof}
Recall that the dimension (over the field of constants) 
of the kernel of a non-degenerate matrix differential operator
is at most the degree of its Dieudonn\'e determinant, see e.g. \cite{DSK13}.
Clearly, the Dieudonn\'e determinants of both $H_0$ and $H_1$ 
have degree 2, therefore their kernels have dimensions (over $\mb F$) at most 2.
On the other hand, obviously 
$\xi^{0,0}\in\ker H_0(\partial)$,
and $\xi^{1,0}\in\ker H_1(\partial)$.
Moreover, the following straightforward identities
\begin{equation}\label{20130305:eq1}
\begin{array}{l}
\displaystyle{
(\partial^3+\partial\circ u+u\partial)\frac1v
+v\partial\big(-\frac u{v^2}-\frac32\frac{(v')^2}{v^4}+\frac{v''}{v^3}\big)=0
\,,}\\
\displaystyle{
\frac{1}{v^2}\partial(u''+4u^2)-\frac1{v^2}Q(\partial)\frac12=0
\,,}
\end{array}
\end{equation}
imply respectively that
$\xi^{0,1}\in\ker H_0(\partial)$,
and $\xi^{1,1}\in\ker H_1(\partial)$,
proving parts (a) and (d).
Parts (b) and (e) follow by straightforward computations.
Note that parts (b) and (e) exactly say, respectively, that the operators $H_0$ and $H_1$ 
satisfy condition (i) of Definition \ref{20130222:def3} of strong skew-adjointness.
We are left to prove condition (ii) for both $H_0$ and $H_1$.
Take $P_0=\Big(\begin{array}{c} p_0 \\ q_0 \end{array}\Big)\perp\ker H_0$.
Since $\tint P_0\cdot\xi^{0,0}=0$ (and since $v$ is invertible in $\mc V$), we have that 
$q_0=(v\alpha_0)'$, for some $\alpha_0\in\mc V$.
Denote $r_0=\frac1v(p_0-\alpha_0'''-(u\alpha_0)'-u\alpha_0')\in\mc V$,
so that $p_0=(\partial^3+\partial\circ u+u\partial)\alpha_0+vr_0$.
The condition $\tint P_0\cdot\xi^{0,1}=0$ then reads
$$
\tint\Big(
r_0+\frac1v(\partial^3+\partial\circ u+u\partial)\alpha_0
+\big(-\frac u{v^2}-\frac32\frac{(v')^2}{v^4}+\frac{v''}{v^3}\big)(v\alpha_0)'\Big)=0\,.
$$
After integration by parts,
the last two terms under integration cancel by the first identity in \eqref{20130305:eq1},
hence the above equation implies that $r_0=\beta_0'\in\partial\mc V$.
In conclusion, $P_0=H_0(\partial)\Big(\begin{array}{c} \alpha_0 \\ \beta_0 \end{array}\Big)$,
proving condition (ii) for $H_0$.
Similarly,
take $P_1=\Big(\begin{array}{c} p_1 \\ q_1 \end{array}\Big)\perp\ker H_1$.
Since $\tint P_1\cdot\xi^{1,0}=0$ (and since $v$ is invertible in $\mc V$), we have that 
$p_1=\big(\frac{\beta_1}{v^2}\big)'$, for some $\beta_1\in\mc V$.
Denote $r_1=v^2q_1+Q(\partial)\frac{\beta_1}{v^2}\in\mc V$,
so that $q_1=\frac{r_1}{v^2}-\frac{1}{v^2}Q(\partial)\frac{\beta_1}{v^2}$.
The condition $\tint P_1\cdot\xi^{1,1}=0$ then reads
$$
\tint\Big(
\big(\frac{\beta}{v^2}\big)'(u''+4u^2)-\frac{1}{2}Q(\partial)\frac{\beta_1}{v^2}
+\frac{r_1}{2}
\Big)=0\,.
$$
After integration by parts,
the first two terms under integration cancel by the second identity in \eqref{20130305:eq1},
hence the above equation implies that $r_1=\alpha_1'\in\partial\mc V$.
In conclusion, $P_1=H_1(\partial)\Big(\begin{array}{c} \alpha_1 \\ \beta_1 \end{array}\Big)$,
proving condition (ii) for $H_1$.
\end{proof}
By Proposition \ref{20130305:prop}(b)-(e), the space of Casimir elements 
for $H_0$ and $H_1$ are respectively
$$
C_{-1}(H_0)=\Span{}_{\mb F}\big\{\tint h^{0,0},\tint h^{0,1}\big\}
\,\text{ and }\,
C_{-1}(H_1)=\Span{}_{\mb F}\big\{\tint h^{1,0},\tint h^{1,1}\big\}\,.
$$
Let, as before, $\{\cdot\,,\,\cdot\}_0$ and $\{\cdot\,,\,\cdot\}_1$ be the Poisson brackets
\eqref{20130222:eq7} associated to the Poisson structures $H_0$ and $H_1$ respectively.
\begin{proposition}\label{20130305:prop2}
The spaces $C_{-1}(H_0)$ and $C_{-1}(H_1)$
are abelian subalgebras with respect to both Poisson brackets 
$\{\cdot\,,\,\cdot\}_0$ and $\{\cdot\,,\,\cdot\}_1$.
\end{proposition}
\begin{proof}
By definition of Casimir elements,
the space $C_{-1}(H_0)$ is in the kernel of the bracket $\{\cdot\,,\,\cdot\}_0$,
and $C_{-1}(H_1)$ is in the kernel of the bracket $\{\cdot\,,\,\cdot\}_1$.
Hence, we only need to prove 
that $C_{-1}(H_0)$ is an abelian subalgebra w.r.t. $\{\cdot\,,\,\cdot\}_1$,
namely, by \eqref{20130222:eq7} and Proposition \eqref{20130305:prop}(b), that
\begin{equation}\label{20130305:eq2}
\tint \xi^{0,0}\cdot H_1(\partial)\xi^{0,1}=0
\,,
\end{equation}
and that $C_{-1}(H_1)$ is an abelian subalgebra w.r.t. $\{\cdot\,,\,\cdot\}_0$,
namely
\begin{equation}\label{20130305:eq3}
\tint \xi^{1,0}\cdot H_0(\partial)\xi^{1,1}=0
\,.
\end{equation}
We have
$$
\begin{array}{l}
\displaystyle{
\vphantom{\Big(}
\xi^{1,0}\cdot H_0(\partial)\xi^{1,1}
=
(\partial^3+\partial\circ u+u\partial)(u''+4u^2)+\frac12 v\partial v^2
} \\
\displaystyle{
\vphantom{\Big(}
=
(u''+4u^2)'''
+\Big(2uu''-\frac12(u')^2+\frac{20}{3}u^3\Big)'
+\frac13 (v^3)'
\,,
}
\end{array}
$$
hence \eqref{20130305:eq3} holds.
A similar, but longer computation, shows that \eqref{20130305:eq2} holds as well.
\end{proof}

\section{The four integrable bi-Hamiltonian equations of lowest order}\label{sec:6}

Let us start computing the first Hamiltonian equations associated to each element
in the kernel of $H_0$ and $H_1$.
Starting with $\xi^{0,0}\in\ker (H_0)$ we get the Hamiltonian vector field
$P^{0,0}=H_1(\partial)\xi^{0,0}$,
and the corresponding Hamiltonian equation (of order 5):
\begin{equation}\label{20130306:eq1}
\begin{array}{rcl}
\frac{du}{dt}
&=&
\Big(\frac1{v^2}\Big)' \\
\frac{dv}{dt}
&=&
-\frac1{v^2}\Big(\frac1{v^2}\Big)^{(5)}
+3\frac1{v^2}\Big(u\Big(\frac1{v^2}\Big)'\Big)''
+3\frac1{v^2}\Big(u\Big(\frac1{v^2}\Big)''\Big)'
\\
&& +2\frac1{v^2}\Big(\frac{u}{v^2}\Big)'''
+2\frac{u}{v^2}\Big(\frac1{v^2}\Big)'''
+8\frac1{v^2}\Big(\frac{u^2}{v^2}\Big)'
+8\frac{u^2}{v^2}\Big(\frac1{v^2}\Big)'
\,.
\end{array}
\end{equation}
Starting with $\xi^{0,1}\in\ker (H_0)$ we get the Hamiltonian vector field
$P^{0,1}=H_1(\partial)\xi^{0,1}$,
and the corresponding Hamiltonian equation (of order 7):
\begin{equation}\label{20130306:eq2}
\begin{array}{rcl}
\frac{du}{dt}
&=&
\Big(-\frac{u}{v^4}+\frac{v''}{v^5}-\frac{3}{2}\frac{(v')^2}{v^6}\Big)'
\\
\frac{dv}{dt}
&=&
-\frac{v'}{v^4}
-\frac1{v^2}
\big(\partial^5+3\partial^2\circ u\partial+3\partial\circ u\partial^2
+2\partial^3\circ u+2u\partial^3
\\
&& +8\partial\circ u^2+8u^2\partial\big)
\Big(-\frac{u}{v^4}+\frac{v''}{v^5}-\frac{3}{2}\frac{(v')^2}{v^6}\Big)
\,.
\end{array}
\end{equation}
Starting with $\xi^{1,0}\in\ker (H_1)$ we get the Hamiltonian vector field
$P^{1,0}=H_0(\partial)\xi^{1,0}$,
and the corresponding Hamiltonian equation:
\begin{equation}\label{20130306:eq3}
\begin{array}{rcl}
\frac{du}{dt} &=& u'
\\
\frac{dv}{dt} &=& v'
\,.
\end{array}
\end{equation}
Finally, starting with $\xi^{1,1}\in\ker (H_0)$ we get the Hamiltonian vector field
$P^{1,1}=H_0(\partial)\xi^{1,1}$,
and the corresponding Hamiltonian equation (of order 5):
\begin{equation}\label{20130306:eq4}
\begin{array}{rcl}
\frac{du}{dt}
&=&
u^{(5)}+10uu'''+25u'u''+20u^2u'+v^2v' \\
\frac{dv}{dt}
&=&
u'''v+u''v'+8uu'v+4u^2v'
\,.
\end{array}
\end{equation}

We want to prove that, for $\epsilon=0,1$, 
each element $\xi^{\epsilon,\alpha},\alpha=0,1$ in the kernel of $H_\epsilon$
produces an infinite Lenard-Magri scheme
starting with $\xi^{\epsilon,\alpha}_0=\xi^{\epsilon,\alpha}$:
\begin{equation}\label{maxi}
\UseTips
\xymatrix{
0 \ar@{<-}[dr]^{H_\epsilon} & & 
P^{\epsilon,\alpha}_{0} \ar@{<-}[dl]_{H_{1-\epsilon}} \ar@{<-}[dr]^{H_\epsilon} & &
P^{\epsilon,\alpha}_{1} \ar@{<-}[dl]_{H_{1-\epsilon}} \ar@{<-}[dr]^{H_\epsilon} & \dots & \in\mc V^2
\\
 & \xi^{\epsilon,\alpha}_{0} & & \xi^{\epsilon,\alpha}_{1} & & \dots & \in\mc V^2
\\
 &\tint h^{\epsilon,\alpha}_{0}\ar@{->}[u]^{\delta}
&&\tint h^{\epsilon,\alpha}_{1}\ar@{->}[u]^{\delta}&& \dots
& \in\tilde{\mc V}/{\partial\tilde{\mc V}}
}
\end{equation}
%
This will easily imply that all 
by-Hamiltonian equations \eqref{20130306:eq1}, \eqref{20130306:eq2}, \eqref{20130306:eq4}
and \eqref{20130306:eq3b}
are integrable, and they are compatible with each other.
Namely, we will prove the following result.
\begin{theorem}\label{20130306:thm}
\begin{enumerate}[(a)]
\item
For $\epsilon,\alpha\in\{0,1\}$ there is a sequence
$\{\tint h^{\epsilon,\alpha}_n\}_{n\in\mb Z_+}\subset\tilde{\mc V}/\partial\tilde{\mc V}$,
such that $\delta(\tint h^{\epsilon,\alpha}_0)=\xi^{\epsilon,\alpha}$ 
and $\delta(\tint h^{\epsilon,\alpha}_n)\in\mc V^2$ for every $n\in\mb Z_+$,
satisfying the Lenard-Magri recurrence relations \eqref{maxi},
i.e.
\begin{equation}\label{20130306:eq5}
H_{1-\epsilon}(\partial)\delta (\tint h^{\epsilon,\alpha}_n)
=H_\epsilon(\partial)\delta (\tint h^{\epsilon,\alpha}_{n+1})
=:\,P^{\epsilon,\alpha}_n
\,.
\end{equation}
\item
The four Lenard-magri schemes are compatible, in the sense that
$$
\begin{array}{l}
\displaystyle{
\vphantom{\Big(}
\big\{\tint h^{\epsilon,\alpha}_m,\tint h^{\delta,\beta}_n\big\}_\zeta=0
\,\,\text{ for all } \epsilon,\delta,\zeta,\alpha,\beta=0,1,\,m,n\in\mb Z_+\,,
} \\
\displaystyle{
\vphantom{\Big(}
[P^{\epsilon,\alpha}_m,P^{\delta,\beta}_n]=0
\,\,\text{ for all } \epsilon,\delta,\alpha,\beta=0,1,\,m,n\in\mb Z_+\,.
}
\end{array}
$$
\item
The differential orders of the higher symmetries 
$P^{\epsilon,\alpha}_n=H_\epsilon(\partial)\delta (\tint h^{\epsilon,\alpha}_n)$ 
tend to infinity as $n\to\infty$.
\end{enumerate}
\end{theorem}
%
%
\begin{remark}\label{20130308:rem1}
Proposition \ref{20130305:prop} gives the integrals of motion $\tint h^{\epsilon,\alpha}_n$
for $n=0$ and arbitrary $\epsilon,\alpha=0,1$.
One can use the Lenard-Magri relations \eqref{20130306:eq5}
to find recursively all other integrals of motions $\tint h^{\epsilon,\alpha}_n$
for arbitrary $n\in\mb Z$.
For example, we have
$$
\tint h^{1,0}_{1}=\tint\Big(\frac13 uv^3+\frac83 u^4+\frac12uu^{(4)}-6u(u')^2\Big)\,.
$$
The next higher symmetry is $P^{1,0}_1=H_0(\partial)\delta(\tint h^{1,0}_{1})$,
and the corresponding Hamiltonian equation (of order 7) is:
\begin{equation}\label{20130306:eq3b}
\begin{array}{rcl}
\frac{du}{dt}
&=&
(\partial^3+2u\partial+u')(u^{(4)}+12uu''+6(u')^2+\frac{32}{3}u^3+\frac13 v^3) +v\partial(uv^2)
\\
\frac{dv}{dt}
&=&
\partial(vu^{(4)}+12vuu''+6v(u')^2+\frac{32}{3}vu^3+\frac13 v^4)
\,.
\end{array}
\end{equation}
\end{remark}
%

\section{Proof of Theorem \ref{20130306:thm}}\label{sec:7}

Let $\epsilon,\alpha\in\{0,1\}$ and 
let 
\begin{equation}\label{20130307:eq4}
\bigg\{
\xi^{\epsilon,\alpha}_n
=\Big(\begin{array}{l} f^{\epsilon,\alpha}_{n} \\ g^{\epsilon,\alpha}_{n}\end{array}\Big)
\bigg\}_{n=0}^N\subset\mc V^2
\end{equation}
be a finite sequence satisfying
the Lenard-Magri recursive relations
\begin{equation}\label{20130307:eq3}
H_{1-\epsilon}(\partial)\xi^{\epsilon,\alpha}_{n-1}=H_\epsilon(\partial)\xi^{\epsilon,\alpha}_{n}
\end{equation}
for every $n=0,\dots,N$, where $\xi^{\epsilon,\alpha}_{-1}=0$ 
and $\xi^{\epsilon,\alpha}_0=\xi^{\epsilon,\alpha}$.
The main task in the proof of Theorem \ref{20130306:thm} is to show that,
when $\alpha=1$, this sequence can be extended by one step.
This is the content of Corollary \ref{20130307:cor} below,
which is a consequence of the following five lemmas.

\begin{lemma}\label{20130307:lem1a}
\begin{enumerate}[(a)]
\item
For $\epsilon=0$, the Lenard-Magri recursive relations \eqref{20130307:eq3}
translate in the following identities for the entries of $\xi^{0,\alpha}_n$ for $1\leq n\leq N$:
\begin{equation}\label{20130306:eq6}
\begin{array}{l}
\displaystyle{
(vf^{0,\alpha}_n)'
=
\frac1{v^2}(f^{0,\alpha}_{n-1})'-\frac1{v^2} Q(\partial)\frac{g^{0,\alpha}_{n-1}}{v^2}
\,,} \\
\displaystyle{
v(g^{0,\alpha}_{n})'
=
\Big(\frac{g^{0,\alpha}_{n-1}}{v^2}\Big)'-(f^{0,\alpha}_{n})'''-(uf^{0,\alpha}_{n})'-u(f^{0,\alpha}_{n})'
\,.}
\end{array}
\end{equation}
\item
For $\epsilon=1$, equations \eqref{20130307:eq3}
translate in the following identities for the entries of $\xi^{1,\alpha}_n$ for $1\leq n\leq N$:
\begin{equation}\label{20130306:eq7}
\begin{array}{l}
\displaystyle{
\Big(\frac{g^{1,\alpha}_{n}}{v^2}\Big)'
=
(f^{1,\alpha}_{n-1})'''+(uf^{1,\alpha}_{n-1})'+u(f^{1,\alpha}_{n-1})'+v(g^{1,\alpha}_{n-1})'
\,,} \\
\displaystyle{
(f^{1,\alpha}_{n})'
=
v^2(vf^{1,\alpha}_{n-1})'+Q(\partial)\Big(\frac{g^{1,\alpha}_{n}}{v^2}\Big)
\,.}
\end{array}
\end{equation}
\end{enumerate}
\end{lemma}
\begin{proof}
It follows immediately from the definitions \eqref{20130307:eq1}-\eqref{20130307:eq2}
of the operators $H_0$ and $H_1$.
\end{proof}

The algebra of differential functions $\mc V$, defined by \eqref{20130307:eq5},
admits the following two differential subalgebras:
\begin{equation}\label{20130307:eq6}
\mc V^+=\mc R_2=\mb F[u,v,u',v',u'',v'',\dots]
\,\,,\,\,\,\,
\mc V^-=\mb F\big[u,\frac1v,u',v',u'',v'',\dots\big]
\,,
\end{equation}
whose intersection is the differential subalgebra
\begin{equation}\label{20130307:eq7}
\mc V^0=\mb F[u,u',v',u'',v'',\dots]\,.
\end{equation}
\begin{lemma}\label{20130307:lem4}
\begin{enumerate}[(a)]
\item
We have $\partial\mc V\cap\mc V^+=\partial\mc V^+$.
\item
We have 
$\partial\mc V\cap\mc V^-=\partial(\mb F v\oplus\mc V^-)$,
\item
For every $k\geq1$, we have
$\partial\mc V\cap\big(\mb F\oplus\frac1{v^k}\mc V^-\big)
=\partial\big(\mb F\frac1{v^{k-1}}\oplus\frac1{v^k}\mc V^-\big)$.
\end{enumerate}
\end{lemma}
\begin{proof}
Any element $f\in\mc V$ admits a unique decomposition
$f=\sum_{j=M}^Nv^jc_j$, 
where $M\leq N\in\mb Z$, $c_j\in\mc V^0$ for all $j$ and $c_M,c_N\neq0$.
Its derivative is then
$$
f'=
\sum_{j=M-1}^{N-1}(j+1)v^jc_{j+1}v'
+\sum_{j=M}^Nv^jc_j'\,.
$$
If $f'\in\mc V^+$, then $M$ must be non-negative, i.e. $f\in\mc V^+$, proving part (a).
Suppose next that $f'\in\mc V^-$. 
If $N\geq2$, then $c_N'=0$, $c_{N-1}'+Nc_Nv'=0$,
namely $0\neq c_N\in\mb F$ and $c_{N-1}+Nc_Nv\in\mb F$,
which is impossible since, by assumption, $c_{N-1}\in\mc V^0$.
If $N=1$, then $c_1'=0$, namely $c_1\in\mb F$,
and therefore $f\in\mb Fv\oplus\mc V^-$.
Finally, if $N\leq0$, then $f\in\mc V^-$.
This completes the proof of part (b).
A similar argument can be used to prove part (c).
Indeed, let $k\geq1$ and assume that $f'\in\mb F\oplus\frac1{v^k}\mc V^-$.
If $N\geq -k+2$, then $c_N'=0$, $c_{N-1}'+Nc_Nv'=0$
(when $N=0$ or $1$ we are using the fact that $\mb F\cap\partial\mc V=0$),
namely $0\neq c_N\in\mb F$ and $c_{N-1}+Nc_Nv\in\mb F$.
This is impossible since, by assumption, $c_{N-1}\in\mc V^0$.
If $N=-k+1$, then $c_{-k+1}'=0$,
namely $c_{-k+1}\in\mb F$,
and therefore $f\in\mb F\frac1{v^{k-1}}\oplus\frac1{v^k}\mc V^-$.
Finally, if $N\leq k$, then $f\in\frac1{v^k}\mc V^-$.
\end{proof}
\begin{lemma}\label{20130307:lem1b}
\begin{enumerate}[(a)]
\item
If $\epsilon=0$, 
we have $f^{0,\alpha}_n\in\frac1v \mc V^-$
and $g^{0,\alpha}_n\in\mb F\oplus\frac1v \mc V^-$
for every $n=0,\dots,N$.
\item
If $\epsilon=1$,
we have $f^{1,\alpha}_n\in\mc V^+$
and $g^{1,\alpha}_n\in v^2 \mc V^+$
for every $n=0,\dots,N$.
\end{enumerate}
\end{lemma}
\begin{proof}
First, let us prove part (a) by induction on $n\geq0$.
The claim clearly holds for $n=0$ by \eqref{20130307:eq10}.
The inductive assumption, together with the first equation in \eqref{20130306:eq6},
implies that
$(vf^{0,\alpha}_n)'\in\frac1v \mc V^-$.
Therefore $f^{0,\alpha}_n\in\frac1v \mc V^-$ thanks to Lemma \ref{20130307:lem4}(c)
for $k=1$.
Furthermore, 
the second equation in \eqref{20130306:eq6} implies that
$(g^{0,\alpha}_n)'\in\frac1v\mc V^-$,
so that, by Lemma \ref{20130307:lem4}(c) with $k=1$, 
we get that $g^{0,\alpha}_n\in\mb F\oplus\frac1v\mc V^-$.
Similarly, we prove part (b) again by induction on $n\geq0$.
For $n=0$ the claim holds by \eqref{20130307:eq11}.
The first equation in \eqref{20130306:eq7} and Lemma \ref{20130307:lem4}(a)
imply that $g^{1,\alpha}_n\in v^2 \mc V^+$,
while the second equation in \eqref{20130306:eq7} and Lemma \ref{20130307:lem4}(a)
imply that $f^{1,\alpha}_n\in \mc V^+$.
\end{proof}

\begin{lemma}\label{20130307:lem2}
For every $n=0,\dots,N$ we have
\begin{equation}\label{20130307:eq23}
\tint \xi^{\epsilon,1-\alpha} \cdot H_{1-\epsilon}(\partial)\xi^{\epsilon,\alpha}_n\in C_{-1}(H_\epsilon)\,.
\end{equation}
\end{lemma}
\begin{proof}
By Proposition \ref{20130305:prop}(b) and (e) and by Theorem \ref{20130222:thm1},
all elements $\xi^{\epsilon,\alpha}_n$ are closed,
and therefore, by Theorem \ref{20130222:thm2}, they are exact in $\tilde{\mc V}$,
i.e. there exist $\tint h^{\epsilon,\alpha}_n\in\tilde{\mc V}/\partial\tilde{\mc V}$
such that $\xi^{\epsilon,\alpha}_n=\delta(\tint h^{\epsilon,\alpha}_n)$ for all $n=0,\dots,N$.
In the Lie superalgebra $\mc W$ of variational polyvector fields we have
(cf. equations \eqref{20130307:eq21} and \eqref{20130307:eq22}):
$$
\tint \xi^{\epsilon,1-\alpha}\cdot H_{1-\epsilon}(\partial)\xi^{\epsilon,\alpha}_n
=[[H_{1-\epsilon},\tint h^{\epsilon,\alpha}_n],\tint h^{\epsilon,1-\alpha}]
\,,
$$
and the assumption \eqref{20130307:eq3} reads
$[H_{\epsilon},\tint h^{\epsilon,\alpha}_{n-1}]=[H_{1-\epsilon},\tint h^{\epsilon,\alpha}_n]$,
for all $n=0,\dots,N$
(where we let $\tint h^{\epsilon,\alpha}_{-1}=0$).
Condition \eqref{20130307:eq23} then
holds by Lemma \ref{20130222:lem}(e).
\end{proof}

\begin{lemma}\label{20130307:lem3}
If $\epsilon\in\{0,1\}$ and $\alpha=1$, we have, for every $n=0,\dots,N$,
\begin{equation}\label{20130307:eq24}
\tint \xi^{\epsilon,0}\cdot H_{1-\epsilon}(\partial)\xi^{\epsilon,1}_n=0\,.
\end{equation}
\end{lemma}
\begin{proof}
For $n=0$ the claim holds by Proposition \ref{20130305:prop2}.
For arbitrary $n\geq1$ we will prove that equation \eqref{20130307:eq24}
holds separately for $\epsilon=0$ and $\epsilon=1$.

Let us consider first the case $\epsilon=0$.
By Lemma \ref{20130307:lem2},
we have that
\begin{equation}\label{20130307:eq25}
\delta\tint \xi^{0,0}\cdot H_{1}(\partial)\xi^{0,1}_n\,\in\,\ker(H_0(\partial))\,.
\end{equation}
Recalling Proposition \ref{20130305:prop}(a),
condition \eqref{20130307:eq25} says that there exist $\alpha,\beta\in\mb F$ such that:
\begin{equation}\label{20130307:eq26}
\delta\tint \xi^{0,0}\cdot H_{1}(\partial)\xi^{0,1}_n=\alpha\xi^{0,0}+\beta\xi^{0,1}\,.
\end{equation}
In other words, recalling the definition \eqref{20130307:eq2} of $H_1$,
we have
\begin{equation}\label{20130307:eq27}
\begin{array}{l}
\displaystyle{
\frac{\delta}{\delta u}
\Big(
\frac1{v^2}(f^{0,1}_n)'-\frac1{v^2}Q(\partial)\frac{g^{0,1}_n}{v^2}
\Big)
=
\beta\frac1v
\,,} \\
\displaystyle{
\frac{\delta}{\delta v}
\Big(
\frac1{v^2}(f^{0,1}_n)'-\frac1{v^2}Q(\partial)\frac{g^{0,1}_n}{v^2}
\Big)
=
\alpha+\beta\Big(-\frac{u}{v^2}+\frac{v''}{v^3}-\frac32\frac{(v')^2}{v^4}\Big)
\,.}
\end{array}
\end{equation}
By Lemma \ref{20130307:lem1b}(a)
both elements $f^{0,1}_n$ and $g^{0,1}_n$ lie in the differential subalgebra $\mc V^-$.
Note that the space $\frac1{v^2}\mc V^-$
is preserved by both $\frac{\delta}{\delta u}$ and $\frac{\delta}{\delta v}$.
It follows that the LHS's of both equations \eqref{20130307:eq27}
lie in $\frac1{v^2}\mc V^-$.
It immediately follows that, necessarily, $\beta=0$ (looking at the first equation),
and $\alpha=0$ (looking at the second equation).
We thus have so far, by equation \eqref{20130307:eq26}, that
$$
\delta\tint \xi^{0,0}\cdot H_{1}(\partial)\xi^{0,1}_n=0
\,.
$$
Recalling that $\ker\delta=\tint \mb F$,
this means that
\begin{equation}\label{20130307:eq28}
\tint \xi^{0,0}\cdot H_{1}(\partial)\xi^{0,1}_n=\tint\gamma\,,
\end{equation}
for some $\gamma\in\mb F$,
and we need to prove that $\gamma=0$.
By writing explicitly equation \eqref{20130307:eq28}, we have that
\begin{equation}\label{20130307:eq29}
\frac1{v^2}(f^{0,1}_n)'-\frac1{v^2}Q(\partial)\frac{g^{0,1}_n}{v^2}
-\gamma\in\partial\mc V\,.
\end{equation}
Lemma \ref{20130307:lem4}(c) with $k=2$
then implies that the LHS of \eqref{20130307:eq29}
lies in $\partial(\mb F\frac1v\oplus\frac1{v^2}\mc V^-)\subset\frac1{v^2}\mc V^-$,
and therefore $\gamma=0$.

For the case $\epsilon=1$ we will use a similar argument.
By Lemma \ref{20130307:lem2},
we have
\begin{equation}\label{20130307:eq25b}
\delta\tint \xi^{1,0}\cdot H_{0}(\partial)\xi^{1,1}_n\,\in\,\ker(H_1(\partial))\,.
\end{equation}
Recalling Proposition \ref{20130305:prop}(d),
condition \eqref{20130307:eq25b} is saying that there exist $\alpha,\beta\in\mb F$ such that:
\begin{equation}\label{20130307:eq26b}
\delta\tint \xi^{1,0}\cdot H_{0}(\partial)\xi^{1,1}_n=\alpha\xi^{1,0}+\beta\xi^{1,1}\,.
\end{equation}
In other words, recalling the definition \eqref{20130307:eq1} of $H_0$,
we have
\begin{equation}\label{20130307:eq27b}
\begin{array}{l}
\displaystyle{
\frac{\delta}{\delta u}
\big(
u(f^{1,1}_n)'+v(g^{1,1}_n)'
\big)
=
\alpha+\beta(u''+4u^2)
\,,} \\
\displaystyle{
\frac{\delta}{\delta v}
\big(
u(f^{1,1}_n)'+v(g^{1,1}_n)'
\big)
=
\frac12 \beta v^2
\,.}
\end{array}
\end{equation}
For $h\in\mc V$, denote by
$D_{h,1}(\partial)=\sum_{n\in\mb Z_+}\frac{\partial h}{\partial u^{(n)}}\partial^n$,
$D_{h,2}(\partial)=\sum_{n\in\mb Z_+}\frac{\partial h}{\partial v^{(n)}}\partial^n$,
its Frechet derivatives,
and by $D_{h,1}(\partial)^*$ and $D_{h,2}(\partial)^*$ the corresponding adjoint operators.
Recalling the definition of the variational derivatives,
equations \eqref{20130307:eq27b} can be equivalently rewritten as follows:
\begin{equation}\label{20130307:eq27c}
\begin{array}{l}
\displaystyle{
(f^{1,1}_n)'-D_{f^{1,1}_n,1}^*(\partial)u'-D_{g^{1,1}_n,1}^*(\partial)v'
=
\alpha+\beta(u''+4u^2)
\,,} \\
\displaystyle{
(g^{1,1}_n)'-D_{f^{1,1}_n,2}^*(\partial)u'-D_{g^{1,1}_n,2}^*(\partial)v'
=
\frac12 \beta v^2
\,.}
\end{array}
\end{equation}
Note that $\xi^{1,1}_n$ is a closed element of $\mc V^2$,
namely it has self-adjoint Frechet derivative.
This means that
$D_{f^{1,1}_n,1}^*(\partial)=D_{f^{1,1}_n,1}(\partial)$,
$D_{g^{1,1}_n,2}^*(\partial)=D_{g^{1,1}_n,2}(\partial)$,
and $D_{f^{1,1}_n,2}^*(\partial)=D_{g^{1,1}_n,1}(\partial)$.
Hence, 
equations \eqref{20130307:eq27c} give
\begin{equation}\label{20130307:eq27d}
\begin{array}{l}
\displaystyle{
(f^{1,1}_n)'-D_{f^{1,1}_n,1}(\partial)u'-D_{f^{1,1}_n,2}(\partial)v'
=
\alpha+\beta(u''+4u^2)
\,,} \\
\displaystyle{
(g^{1,1}_n)'-D_{g^{1,1}_n,1}(\partial)u'-D_{g^{1,1}_n,2}(\partial)v'
=
\frac12 \beta v^2
\,.}
\end{array}
\end{equation}
The polynomial ring $\mc V^+=\mb F[u,v,u',v',u'',v'',\dots]$
admits the polynomial degree decomposition
$\mc V^+=\bigoplus\mc V^+[k]$,
where $\mc V^+$ consists of homogeneous polynomials of degree $k$.
For $h\in\mc V^+$, denote by $h=\sum_{k\in\mb Z_+}h[k]$
its decomposition in homogeneous components.
Note that $\mc V^+[0]=\mb F$ and $\partial\mc V\cap\mb F=0$.
It follows, 
by looking at the homogenous components of degree $0$
in both sides of the first equation of \eqref{20130307:eq27d},
that $\alpha=0$.
Furthermore, 
by looking at the homogenous components of degree $2$
in both sides of the second equation of \eqref{20130307:eq27d},
we get
\begin{equation}\label{20130307:eq27e}
(g^{1,1}_n[2])'-D_{g^{1,1}_n[2],1}(\partial)u'-D_{g^{1,1}_n[2],2}(\partial)v'
=
\frac12 \beta v^2\,.
\end{equation}
On the other hand,
by looking at the homogenous components of degree $0$
in both sides of the first equation of \eqref{20130306:eq7},
we get that $g^{1,\alpha}_n[2]\in\mb F v^2$.
But then the LHS of equation \eqref{20130307:eq27e} is equal to zero,
and therefore $\beta=0$.
We thus have so far, by equation \eqref{20130307:eq26b}, that
\begin{equation}\label{20130307:eq28b}
\tint \xi^{1,0}\cdot H_{0}(\partial)\xi^{1,1}_n=\tint\gamma\,,
\end{equation}
for some $\gamma\in\mb F$,
and we need to prove that $\gamma=0$.
By writing explicitly equation \eqref{20130307:eq28b}, we get
\begin{equation}\label{20130307:eq29b}
\gamma=u(f^{1,1}_n)'+v(g^{1,1}_n)'+r'
\,,
\end{equation}
for some $r\in\mc V$.
By Lemma \ref{20130307:lem1b} 
we have that $u(f^{1,1}_n)'+v(g^{1,1}_n)'\in\bigoplus_{k\geq1}\mc V^+[k]$.
Moreover, by Lemma \ref{20130307:lem4}(a) we can assume $r\in\mc V^+$,
and therefore $r'\in\bigoplus_{k\geq1}\mc V^+[k]$.
It follows by equation \eqref{20130307:eq29b} that $\gamma=0$.
\end{proof}

\begin{corollary}\label{20130307:cor}
If $\alpha=1$, there exists $\xi^{\epsilon,1}_{N+1}\in\mc V^2$ solving the equation
$$
H_{1-\epsilon}(\partial)\xi^{\epsilon,1}_N=H_{\epsilon}(\partial)\xi^{\epsilon,1}_{N+1}\,.
$$
\end{corollary}
\begin{proof}
By the usual inductive argument, based on the recursive relations \eqref{20130307:eq3}
and the fact that $H_0$ and $H_1$ are skew-adjoint
(as in the proof of Lemma \ref{20130222:lem}(a)),
we know that $H_{1-\epsilon}(\partial)\xi^{\epsilon,1}_N\perp \xi^{\epsilon,1}_0$.
Moreover, 
Lemma \ref{20130307:lem3} says that 
$H_{1-\epsilon}(\partial)\xi^{\epsilon,1}_N\perp\xi^{\epsilon,0}_0$.
Therefore, 
$H_{1-\epsilon}(\partial)\xi^{\epsilon,1}_N\perp\ker(H_\epsilon)$,
which, by the strong skew-adjointness property of $H_\epsilon$ 
(cf. Proposition \ref{20130305:prop}(c) and (f)),
implies that $H_{1-\epsilon}(\partial)\xi^{\epsilon,1}_N$ lies in the image of $H_\epsilon(\partial)$,
proving the claim.
\end{proof}

\begin{proof}[Proof of Theorem \ref{20130306:thm}]
For $\epsilon\in\{0,1\}$, by Corollary \ref{20130307:cor} there exists
an infinite sequence $\{\xi^{\epsilon,1}_n\}_{n\in\mb Z_+}$
starting with $\xi^{\epsilon,1}_0=\xi^{\epsilon,1}$
and satisfying the Lenard-Magri recursive relations
$H_{1-\epsilon}(\partial)\xi^{\epsilon,1}_{n-1}=H_{\epsilon}(\partial)\xi^{\epsilon,1}_{n}$
for all $n\in\mb Z_+$ (where $\xi^{\epsilon,1}_{-1}=0$).
Next,
we want to prove that the statement of Corollary \ref{20130307:cor}
also holds for $\alpha=0$.
Let
$\{\xi^{\epsilon,0}_n\}_{n=0}^N$
be a finite sequence starting with $\xi^{\epsilon,0}_0=\xi^{\epsilon,0}$
and satisfying the Lenard-Magri recursive relations
$H_{1-\epsilon}(\partial)\xi^{\epsilon,0}_{n-1}=H_{\epsilon}(\partial)\xi^{\epsilon,0}_{n}$
for all $n=0,\dots,N$ (where $\xi^{\epsilon,\alpha}_{-1}=0$).
By the usual inductive argument
we know that $H_{1-\epsilon}(\partial)\xi^{\epsilon,0}_N\perp\xi^{\epsilon,0}_0$.
Moreover, we have, by the Lenard-Magri relations and skew-adjointness of $H_0$ and $H_1$,
$$
\begin{array}{l}
\displaystyle{
\vphantom{\Big(}
\tint \xi^{\epsilon,1}_0\cdot H_{1-\epsilon}(\partial)\xi^{\epsilon,0}_N
= -\tint \xi^{\epsilon,0}_N\cdot H_{1-\epsilon}(\partial) \xi^{\epsilon,1}_0
= -\tint \xi^{\epsilon,0}_N\cdot H_{\epsilon}(\partial) \xi^{\epsilon,1}_1
} \\
\displaystyle{
\vphantom{\Big(}
= \tint \xi^{\epsilon,1}_1\!\cdot\! H_{\epsilon}(\partial) \xi^{\epsilon,0}_N
= \tint \xi^{\epsilon,1}_1\!\cdot\! H_{1-\epsilon}(\partial) \xi^{\epsilon,0}_{N-1}
=\,\dots\,
= \tint \xi^{\epsilon,1}_{N+1}\!\cdot\! H_{1-\epsilon}(\partial) \xi^{\epsilon,0}_{-1}=0
.
} 
\end{array}
$$
Hence, 
$H_{1-\epsilon}(\partial)\xi^{\epsilon,0}_{n-1}\perp\ker(H_\epsilon)$,
and therefore, by the same argument as in the proof of Corollary \ref{20130307:cor},
there exists
$\xi^{\epsilon,0}_{N+1}\in\mc V^2$ solving
the equation $H_{1-\epsilon}(\partial)\xi^{\epsilon,0}_N=H_{\epsilon}(\partial)\xi^{\epsilon,0}_{N+1}$.
Therefore, the given finite sequence can be extended
to an infinite sequence $\{\xi^{\epsilon,0}_n\}_{n\in\mb Z_+}$
satisfying the Lenard-Magri recursive relations.

So far, for each $\epsilon,\alpha\in\{0,1\}$, we have an infinite sequence
$\{\xi^{\epsilon,\alpha}_n\}_{n\in\mb Z_+}$
satisfying the Lenard-Magri recursive relations \eqref{20130307:eq3}.
By Theorems \ref{20130222:thm1} and \ref{20130222:thm2}
we know that all the elements of these sequences are exact in $\tilde{\mc V}$,
i.e. there are elements $\tint h^{\epsilon,\alpha}_n\in\tilde{\mc V}/\partial\tilde{\mc V}$
such that $\delta h^{\epsilon,\alpha}_n=\xi^{\epsilon,\alpha}_n$
for all $\epsilon,\alpha\in\{0,1\},\,n\in\mb Z_+$.
and the relations \eqref{20130307:eq3} on the elements $\xi^{\epsilon,\alpha}_n$'s
translate to the equations \eqref{20130306:eq5} on the elements $\tint h^{\epsilon,\alpha}_n$'s.
Part (a) of the theorem is then proved, and part (b) is an immediate consequence 
of Lemma \ref{20130222:lem}(a), (b) and (c).

We are left to prove part (c).
By looking at the recursive equations \eqref{20130306:eq6}
it is not hard to compute the differential orders $|f^{\epsilon,\alpha}_n|$
and $|g^{\epsilon,\alpha}_n|$ of all the entries of each element of the four sequences.
We have for every $n\geq1$:
$$
\begin{array}{l}
\displaystyle{
\vphantom{\Big(}
|f^{0,0}_n|=6n-2\,,\,\,
|g^{0,0}_n|=6n\,,\,\,
|f^{0,1}_n|=6n\,,\,\,
|g^{0,1}_n|=6n+2\,,
} \\
\displaystyle{
\vphantom{\Big(}
|f^{1,0}_n|=6n-2\,,\,\,
|g^{1,0}_n|=6(n-1)\,,\,\,
|f^{1,1}_n|=6n+2\,,\,\,
|g^{1,1}_n|=6n-2\,.
}
\end{array}
$$
Hence, the differential orders of the higher symmetries are
$|P^{0,0}_n|=6n+5,\,|P^{0,1}_n|=6n+7,\,|P^{1,0}_n|=6n+1,\,|P^{1,1}_n|=6n+5$.
Claim (c) follows.
\end{proof}

\begin{remark}
Since the algebra of differential polynomials $\mc V^+$ is normal,
it follows from Lemma \ref{20130307:lem1b}(b)
that all integrals of motion $\tint h^{1,\alpha}_n$ lie in $\mc V^+/\partial\mc V^+$.
\end{remark}

\begin{remark}
From equation \eqref{20130306:eq6} and Lemma \ref{20130307:lem4}(c)
it is easy to see that, for every $\alpha\in\{0,1\}$ and $n\geq0$, we have
$f^{0,\alpha}_n\in\mb F\frac1v\oplus\frac1{v^2}\mc V^-$,
$g^{0,\alpha}_n\in\frac1{v^2}\mc V^-$.
It follows from the arguments in the proof of \cite[Thm.3.2]{BDSK09}
that all integrals of motion $\tint h^{0,\alpha}_n$ lie in $\mc V^-/\partial\mc V^-$.
\end{remark}


\end{document}